\documentclass[manyauthors]{fundam}
\usepackage{url} % takes care of hyperlinks, preferred over hyperref
\usepackage[ruled,lined]{algorithm2e}% provides Algorithm environment
\usepackage{graphicx}% allows for inclusion of graphic files (figures)

\begin{document}

\setcounter{page}{357}
\publyear{22}
\papernumber{2114}
\volume{185}
\issue{4}

  \finalVersionForARXIV
 %% \finalVersionForIOS

\title{A Local Diagnosis Algorithm for Hypercube-like Networks \\ under the BGM Diagnosis Model}

\author{Cheng-Kuan Lin\\
Department of Computer Science\\
National Yang Ming Chiao Tung University \\
Hsinchu City, Taiwan 30010, R.O.C.
\and Tzu-Liang Kung\\
Department of Computer Science and \\
Information Engineering, Asia University \\
Taichung City, Taiwan 413, R.O.C.
\and Chun-Nan Hung \\
Department of Information Management \\
Da-Yeh University\\
 Changhua, Taiwan 51591, R.O.C.
\and Yuan-Hsiang Teng\thanks{Address for correspondence: Department of Computer
             Science and Information Engineering, Providence University, Taichung City, Taiwan 433, R.O.C. \newline \newline
          \vspace*{-6mm}{\scriptsize{Received July 2021; \ accepted April 2022.}}}
\\
Department of Computer Science and\\
 Information Engineering, Providence University \\
Taichung City, Taiwan 433, R.O.C.\\
yhteng@pu.edu.tw
}

\maketitle

\runninghead{C.K. Lin et al.}{A Diagnosis Algorithm for Hypercube-like Networks under BGM Model}

\begin{abstract}
System diagnosis is process of identifying faulty nodes in a system. An efficient diagnosis is crucial for a multiprocessor system. The BGM diagnosis model is a modification of the PMC diagnosis model, which is a test-based diagnosis. In this paper, we present a specific structure and propose an algorithm for diagnosing a node in a system under the BGM model. We also give a polynomial-time algorithm that a node in a hypercube-like network can be diagnosed correctly in three test rounds under the BGM diagnosis model.
\end{abstract}


\section{Introduction}\label{sec1}

Sensor networks are being increasingly used in computer technology. Among various communication networks, sensor networks have been widely used in many fields, such as military, environmental monitoring, smart agriculture, healthcare, traffic systems, etc. Sensors in the network are usually created by incorporating sensing materials with integrated circuits. Sensor networks can be established by enhancing each sensor node with wireless communication capabilities and networking the sensor nodes together \cite{Intanagonwiwat2000,Kung2016}. Since sensor networks essentially have no infrastructure, it is more feasible to manage and update the topology of certain substructure rather than the entire network. The substructure, for example, may be a ring, a path, a tree, a mesh, etc. Sensor networks have multiple sensor nodes or processors. In general, sensor nodes, processors, and links are modeled as a graph topology. Even slight malfunctions can render the service ineffective; thus, system reliability is a crucial parameter that must be considered when designing sensor networks. To ensure system reliability, faulty devices should be replaced with fault-free ones immediately. The fault here means that the device has a calculation error or a problem with the sensing information. At first glance, the device can operate normally, but there is a problem with the calculation result or the sensing information. For example, the communication of a sensor node is normal, but the sensing data may be abnormal. This situation is quite common in wireless sensor networks, but it may also happen in multi-processing systems. Hence the faulty device we discussed is not that the device is completely inoperable, or that the communication system is broken such that the device cannot exchange information with other devices in the network. The fault type in this article is a partial fault. That is, the communication is normal, but the calculation or sensing data is abnormal. {\it System diagnosis} refers to the process of identifying faulty devices. The maximum number of faulty devices that can be identified accurately is called the {\it diagnosability} of a system. If all the faulty devices in a system can be detected precisely and the maximum number of faulty devices is $t$, then the system is {\it $t$-diagnosable}. Many studies on system diagnosis and diagnosability have been reported \cite{Chang2014,Lin2015,Teng2014,Li2017,Zhu19,Lin2020}.

\medskip
Barsi, Grandoni, and Maestrini proposed the BGM diagnosis model in \cite{BGM76}. The BGM model is a test-based diagnosis and a modification of the PMC diagnosis model presented by Preparata et al. \cite{PMC5}. Under the BGM model, a processor diagnoses a system by testing the neighboring processors via the links between them. A few related studies have investigated the BGM model \cite{Albini2004,Blough1999,Vedeshenkov2002}. In this paper, we propose a sufficient and necessary characterization of a $t$-diagnosable system under the BGM diagnosis model. We present a specific structure named the $t$-diagnosis-tree, and propose an algorithm for diagnosing nodes in a system. With our algorithm, the faulty or fault-free status of a node can be identified accurately if the total number of faulty vertices does not exceed $t$, the connectivity of the system. We also discuss the conditional local diagnosis of the BGM model and present an algorithm for diagnosing a system with a conditionally faulty set.

\medskip
The hypercube is one of the commonest topologies in all interconnection networks appeared in the literature \cite{Leighton92}. The properties of the hypercube have been studied for many years. The hypercube has remained eye-catching to this day. By twisting certain pairs of links in the hypercube, many different network structures are presented \cite{Abraham91,Cull95,Efe91,Efe92}. To make a unified study of these variants, Vaidya et al. introduced the class of hypercube-like graphs in \cite{Vaidya93}. The hypercube-like networks, consisting of simple, connected, and undirected graphs, contain most of the hypercube variants. In this paper, we prove that the nodes in an $n$-dimensional hypercube-like network $XQ_n$ can be diagnosed correctly with a faulty set $F$ in three test rounds under the BGM diagnosis model if $|F|\leq n$. The remainder of this paper is organized as follows. In Section \ref{sec2}, we introduce the BGM diagnosis model. In Section \ref{sec3}, we give some properties about the $n$-dimensional hypercube-like graphs $XQ_n$. In Section \ref{sec4}, we propose a specific structure for local diagnosis, and present a local diagnosis algorithm for the BGM diagnosis model. In Section \ref{sec5}, we propose an algorithm for conditional local diagnosis under the BGM model. We give a 3-round local diagnosis algorithm for the hypercube-like network under the BGM diagnosis model in Section \ref{sec6}. Finally, Section \ref{sec7} presents our conclusions.

\section{The BGM diagnosis model}\label{sec2}

%\section{Preliminaries}\label{sec2}

We use an interconnection network to represent the layout of processors and links in a high-speed multiprocessor system. An interconnection network is typically modeled as an undirected graph in which the nodes represent processors and the edges represent the communication links between the processors. For graph definitions and notations, we follow \cite{HsuBook}. Let $G=(V,E)$ be a {\it graph} if $V$ is a finite set and $E$ is a subset of \{$\{u,v\} \mid \{u,v\}$ is an unordered pair of $V$\}. We define $V$ as the {\it node set} and $E$ as the {\it edge set} of $G$. Two nodes $u$ and $v$ are {\it adjacent} if $\{u,v\}\in E$; we define $u$ as a {\it neighbor} of $v$, and vice versa. We use $N_G(u)$ to represent the neighborhood set $\{v \mid \{u,v\} \in E(G)\}$. The {\it degree} of a node $v$ in a graph $G$, represented as $\textrm{deg}_{G}(v)$, is the number of edges incident to $v$.

\medskip\smallskip
For the specialized terms of the BGM diagnosis model, we follow \cite{BGM76}. We assume that adjacent processors can perform tests on each other under this model. Let $G=(V,E)$ denote the underlying topology of a multiprocessor system. For any two adjacent nodes $u,v\in V(G)$, the ordered pair $(u,v)$ represents a {\it test} in which processor $u$ can diagnose processor $v$. In this situation, $u$ is a {\it tester} and $v$ is a {\it testee}. If $u$ evaluates $v$ as faulty, the result of the test $(u,v)$ is $1$; if $u$ evaluates $v$ as fault-free, the result of the test $(u,v)$ is $0$. Because we consider the faults permanent, the result of a test is {\it reliable} if and only if the tester is fault-free. To make a system with a complex structure more realistic, we assume that every processor has computational ability. Therefore, any completed test for a given set of faults in a processor consists of a sequence of numerous stimuli. Based on the observation, it is reasonable to suppose that, between real and expected reaction to the stimuli, at least one mismatch will occur as long as the tested processor is faulty, even if the testing processor is faulty. Following this discussion, the diagnostic model for a system $G$ is defined as follows. Assume that $x$ and $y$ are two adjacent processors in $G$. If $x$ is fault-free, the result of the test $(x,y)$ is $0$ if $y$ is fault-free and $1$ if $y$ is faulty. If $x$ is faulty and $y$ is fault-free, both results of are possible. If $x$ and $y$ are faulty, the result of the test $(x,y)$ is $1$ (Table \ref{BGMtable}). According to the aforementioned definition, if a test result is $0$, the testee should definitely be fault-free. By contrast, if the test result is $1$, a fault exists in the tester, testee, or both.

\begin{table}[h]
\renewcommand{\arraystretch}{1.3}
 \caption{Results of $\sigma(x,y)$ and $\sigma(y,x)$ under the BGM diagnosis model.} \label{BGMtable}
 \centering
\begin{tabular}{|c|c|c|c|} \hline
Node $x$ & Node $y$ & $\sigma(x,y)$ & $\sigma(y,x)$ \\ \hline
Fault-free & Fault-free & 0 & 0\\ \hline
Fault-free & Faulty & 1 & 0 or 1\\ \hline
Faulty & Fault-free & 0 or 1 & 1\\ \hline
Faulty & Faulty & 1 & 1\\
\hline
\end{tabular}
\end{table}

\eject
A {\it test assignment} for system $G$ is a collection of tests that can be modeled as a directed graph $T=(V,L)$. Thus, $(u,v)\in L$ means that $u$ and $v$ are adjacent in $G$. The collection of all test results from the test assignment $T$ is called a {\it syndrome}. Formally, a syndrome of $T$ is a mapping $\sigma:L\rightarrow \{0,1\}$. A {\it faulty set} $F$ is the set of all faulty processors in $G$. Note that $F$ can be any subset of $V$. System diagnosis is the process of identifying faulty nodes in a system. The maximum number of faulty nodes that can be accurately identified in a system $G$ is called the {\it diagnosability} of $G$, denoted by $\tau(G)$. A system $G$ is {\it $t$-diagnosable} if all faulty nodes in $G$ can be precisely detected with the total number of faulty nodes being at most $t$. Let $\sigma$ denote the syndrome resulting from a test assignment $T=(V,L)$. A subset of nodes $F\subseteq V$ is considered {\it consistent} with $\sigma$ if for a $(u,v)\in L$ such that $u\in V-F$, $\sigma(u,v)=1$ if and only if $v\in F$. Let $\sigma(F)$ denote the set of all possible syndromes with which the faulty set $F$ can be consistent. Two distinct faulty sets $F_1$ and $F_2$ of $V$ are {\it distinguishable} if $\sigma(F_1)\cap\sigma(F_2)=\emptyset$; otherwise, $F_1$ and $F_2$ are {\it indistinguishable}. Thus, $(F_1,F_2)$ is a {\it distinguishable pair} of faulty sets if $\sigma(F_1)\cap\sigma(F_2)=\emptyset$; otherwise, $(F_1,F_2)$ is an {\it indistinguishable pair}. For any two distinct faulty sets $F_1$ and $F_2$ of $G$ with $|F_1|\leq t$ and $|F_2|\leq t$, a system $G$ is $t$-diagnosable if and only if $(F_1,F_2)$ is a distinguishable pair. Let $F_1$ and $F_2$ be two distinct sets. We use $F_1\triangle F_2$ to denote the symmetric difference $(F_1-F_2)\cup (F_2-F_1)$ between $F_1$ and $F_2$. Many researchers study the conventional diagnosability that describes the global status of a system under the random-fault model. Thus, Hsu and Tan proposed the concept of local diagnosability in \cite{Hsu2007}. The research about local diagnosability concerns with the local connective substructure in a system. Some related studies have been proposed in \cite{Chiang2012,Cheng2013,Teng2013,Wang2018}. Suppose that $\sigma_{F}$ is a syndrome produced by a set of faulty nodes $F\subseteq V$ containing $u$ with $|F|\leq t$. We consider $G$ {\it locally $t$-diagnosable} at $u$ if every faulty node set $F'$ compatible with $\sigma_{F}$ and $|F'|\leq t$ also contains $u$. The {\it local diagnosability} of $u$ is the maximum value of $t$ such that $G$ is locally $t$-diagnosable at $u$.

\section{The hypercube-like graphs}\label{sec3}

Let $G_0=(V_0,E_0)$ and $G_1=(V_1,E_1)$ be two disjoint graphs with the same number of nodes. A {\it 1-1 connection} between $G_0$ and $G_1$ is defined as an edge set $E=\{(v,\phi(v))\mid v\in V_0,\phi(v)\in V_1,$ and $\phi$ : $V_0\rightarrow V_1$ is a bijection $\}$. We use $G_0\oplus G_1$ to denote $G=(V_0\cup V_1, E_0\cup E_1\cup E)$. The operation "$\oplus$" may generate different graphs depending on the bijection $\phi$. There are some studies on the operation "$\oplus$". Let $G=G_0\oplus G_1$, and let $x$ be any node in $G$. We use $\bar{x}$ to denote the unique node matched under $\phi$.

\medskip
Now, we can define the set of $n$-dimensional hypercube-like graph $XQ_n$ as follows: \medskip

(1) $XQ_1=\{K_2\}$, where $K_2$ is the complete graph with two nodes.

(2) Assume that $G_0\in XQ_n$ and $G_1\in XQ_n$. Then $G=G_0\oplus G_1$ is a graph in $XQ_{n+1}$.

\medskip
Every graph in $XQ_n$ is an $n$-regular graph with $2^n$ nodes. Let $G$ be a graph in $XQ_{n+1}$. Then $G=G_0\oplus G_1$ with both $G_0$ and $G_1$ in $XQ_n$. Suppose that $u$ is a node in $V(G)$. Then $u$ is a node in $V(G_i)$ for some $i\in \{0,1\}$. We use $\bar{u}$ to denote the node in $V(G_{1-i})$ matched under $\phi$. The $1$-dimensional hypercube-like graph $XQ_1$ is a complete graph with two nodes and the edge is labeled by $1$. An $n$-dimensional hypercube-like graph $XQ_n$ can be generated by two $(n-1)$-dimensional hypercube-like graphs, denoted $XQ^{0}_{n-1}$ and $XQ^{1}_{n-1}$, and a perfect match between the nodes of $XQ^{0}_{n-1}$ and $XQ^{1}_{n-1}$, where every edge in this perfect match is labeled by $n$. The following is some properties about the $n$-dimensional hypercube-like graphs.

\begin{theorem}\label{neighborXQ}\cite{Fan05}
Let $n$ and $k$ be any two integers with $n\geq 3$ and $1\leq k\leq 2n-2$. For any node subset $V$ of $XQ_{n}$ with $|V|=k$, $|N_{XQ_{n}}(V)|\geq kn-\frac{k(k+1)}{2}+1$.
\end{theorem}

\begin{lemma}\label{l1}
Let $\{u_1,u_2,\ldots,u_k\}$ be a set of $k$ distinct nodes of $XQ_n$ with $k\leq n$. There exists a node set $\{v_1,v_2,\ldots,v_k\}$ in $XQ_{n}-\{u_1,u_2,\ldots,u_k\}$ such that $(u_i,v_i)\in E(XQ_n)$ for every $1\leq i\leq k$.
\end{lemma}

\begin{proof}
We prove the lemma by induction on $n$. For $n=1$, the lemma follows trivially. Suppose that the lemma holds on $XQ_m$ for $1\leq m\leq n-1$. Without loss of generality, we assume that $k=n$. Let $A_i=\{u_1,u_2,\ldots,u_n\}\cap V(XQ^{i}_{n-1})$ for $i\in \{0,1\}$. Without loss of generality, we assume that $|A_0|\geq |A_1|$. We have the following cases. \\

{\bf Case 1:} Suppose that $|A_0|\leq n-1$. Since $|A_0|\geq |A_1|$, $|A_1|\leq n-1$. Without loss of generality, we assume that $A_0=\{u_1,u_2,\ldots,u_t\}$ and $A_1=\{u_{t+1},u_{t+2},\ldots,u_n\}$. By induction, there exists a node set $\{v_1,v_2,\ldots,v_t\}$ in $XQ^{0}_{n-1}-\{u_1,u_2,\ldots,u_t\}$ such that $(u_i,v_i)\in E(XQ^{0}_{n-1})$ for every $1\leq i\leq t$, and there exists node set $\{v_{t+1},v_{t+2},\ldots,v_{n}\}$ in $XQ^{1}_{n-1}-\{u_{t+1},u_{t+2},\ldots,u_{n}\}$ such that $(u_i,v_i)\in E(XQ^{1}_{n-1})$ for every $t+1\leq i\leq n$. Thus $\{v_1,v_2,\ldots,v_n\}$ forms a desired set. \medskip

{\bf Case 2:} Suppose that $|A_0|=n$. For every $1\leq i\leq n$, we set $v_i$ being the neighbor of $u_i$ where $(u_i,v_i)$ is labeled by $n$. Thus $\{v_1,v_2,\ldots,v_n\}$ forms a desired set. \medskip

Thus the lemma holds.
\end{proof}

\begin{lemma}\label{l2}
Suppose that $n\geq 4$. Let $\{u_1,u_2,\ldots,u_{n+1}\}$ be a set of $n+1$ distinct nodes of $XQ_n$, where it is not isomorphic to $K_{1,n}$. Then there exists a node set $\{v_1,v_2,\ldots,v_{n+1}\}$ of $XQ_n-\{u_1,u_2,\ldots,u_{n+1}\}$ such that $(u_i,v_i)\in E(XQ_{n})$ for every $1\leq i\leq n+1$.
\end{lemma}

\begin{proof}
We set $A_i=\{u_1,u_2,\ldots,u_{n+1}\}\cap V(XQ_{n-1}^{i})$ for $i\in \{0,1\}$. Without loss of generality, we assume that $|A_0|\geq |A_1|$. Then we have the following cases. \medskip

\noindent{\bf Case 1:} Suppose that $|A_0|\leq n-1$. Since $|A_0|\leq n-1$, $|A_1|\leq n-1$. Without loss of generality, we assume that $A_0=\{u_1,u_2,\ldots,u_{t}\}$ and $A_1=\{u_{t+1},u_{t+2},\ldots,u_{n+1}\}$. By Lemma \ref{l1}, there exists a $t$ distinct node set $\{v_1,v_2,\ldots,v_{t}\}$ of $XQ_{n-1}^{0}-\{u_1,u_2,\ldots,u_{t}\}$ such that $(u_i,v_i)\in E(XQ^{0}_{n-1})$ for every $1\leq i\leq t$. Similarly, there exists a $n-t+1$ distinct node set $\{v_{t+1},v_{t+2},\ldots,v_{n+1}\}$ of $XQ_{n-1}^{1}-\{u_{t+1},u_{t+2},\ldots,u_{n+1}\}$ such that $(u_i,v_i)\in E(XQ^{1}_{n-1})$ for every $t+1\leq i\leq n+1$. Thus $\{v_1,v_2,\ldots,v_{n+1}\}$ forms a desired set. \medskip

\noindent{\bf Case 2:} Suppose that $|A_0|=n$. Without loss of generality, we assume that $A_0=\{u_1,u_2,\ldots,u_{n}\}$ and $A_1=\{u_{n+1}\}$.

\medskip
\noindent{\bf Case 2.1:} Suppose that $(u_i,u_{n+1})\notin E(XQ_n)$ for every $1\leq i\leq n$. For every $1\leq i\leq n+1$, we set $v_i$ being the neighbor of $u_i$ where $(u_i,v_i)$ is labeled by $n$. Thus $\{v_1,v_2,\ldots,v_{n+1}\}$ forms a desired set. \medskip

\noindent{\bf Case 2.2:} Suppose that $(u_i,u_{n+1})\in E(XQ_n)$ for some $1\leq i\leq n$. Without loss of generality, we assume that $(u_1,u_{n+1})\in E(XQ_n)$. Since $\{u_1,u_2,\ldots,u_{n+1}\}$ is not isomorphic to $K_{1,n+1}$ and $(u_1,u_{n+1})\in E(XQ_n)$, $(u_1,u_i)\notin E(XQ_n)$ for some $2\leq i\leq n$. Since $deg_{XQ^{0}_{n-1}}(u_1)=n-1$ and $(u_1,u_i)\notin E(XQ_n)$ for some $2\leq i\leq n$, there exists a node $v_1\in V(XQ_{n-1}^{0})-\{u_1,u_2,\ldots,u_{n}\}$ such that $(u_1,v_1)\in E(XQ_{n-1}^{0})$. By Theorem \ref{neighborXQ}, $|N(\{u_2,u_3,\ldots,u_n\})|\geq (n-1)(n-1)-\frac{(n-1)n}{2}+1=\frac{(n^2-3n)}{2}+2\geq 3$ if $n\geq 4$. Thus there exists a node $v_2\in V(XQ_{n-1}^{0})-\{u_1,u_2,\ldots,u_{n},v_1\}$ such that $(u_i,v_2)\in E(XQ_{n-1}^{0})$ for some $2\leq i\leq n$. Without loss of generality, we assume that $(u_2,v_2)\in E(XQ_{n-1}^{0})$. For every $3\leq i\leq n$, let $v_i$ be the neighbor of $u_i$ where $(u_i,v_i)$ is labeled by $n$. Since $deg_{XQ^{1}_{n-1}}(u_{n+1})=n-1$, there exists a node $v_{n+1}\in V(XQ_{n-1}^{1})-\{v_3,v_4,\ldots,v_{n}\}$ such that $(u_{n+1},v_{n+1})\in E(XQ_{n-1}^{1})$. Thus $\{v_1,v_2,\ldots,v_{n+1}\}$ forms a desired set. \medskip

\noindent{\bf Case 3:} Suppose that $|A_0|=n+1$. For every $1\leq i\leq n+1$, we set $v_i$ being the neighbor of $u_i$ where $(u_i,v_i)$ is labeled by $n$. Thus $\{v_1,v_2,\ldots,v_{n+1}\}$ forms a desired set.
\end{proof}

\begin{lemma}\label{l3}
Suppose that $n\geq 3$. Let $\{u_1,u_2,\ldots,u_{n+1}\}$ be a set of $n+1$ distinct nodes of $XQ_n$, where $N(u_1)=\{u_2,u_3,\ldots,u_{n+1}\}$. Then there exists a node set $\{v_2,v_3,\ldots,v_{n+1}\}$ of $XQ_n-\{u_1,u_2,\ldots,u_{n+1}\}$ such that $(u_i,v_i)\in E(XQ_{n})$ for every $2\leq i\leq n+1$.
\end{lemma}

\begin{proof}
We prove the lemma by induction on $n$. Suppose that $n=3$. The desired set is illustrated in Figure \ref{xq3_proof}. Suppose that the lemma holds on $XQ_m$ for $3\leq m\leq n-1$. Without loss of generality, we assume that $u_{n+1}$ is the neighbor of $u_1$ where $(u_1,u_{n+1})$ is labeled by $n$. Thus, we have $\{u_1,u_2,\ldots,u_{n}\}\subseteq V(XQ_{n-1}^{j})$ and $\{u_{n+1}\}\subseteq V(XQ_{n-1}^{1-j})$ for some $j\in \{0,1\}$. By induction, there is a node set $\{v_2,v_3,\ldots,v_{n}\}$ of $XQ_{n-1}^{j}-\{u_1,u_2,\ldots,u_{n}\}$ such that $(u_i,v_i)\in E(XQ_{n-1}^{j})$ for every $2\leq i\leq n$. Let $v_{n+1}$ be the neighbor of $u_{n+1}$ where $(u_{n+1},v_{n+1})$ is labeled by $2$. Obviously, $v_{n+1}\in V(XQ_{n-1}^{1-j})$. Thus $\{v_2,v_3,\ldots,v_{n+1}\}$ forms a desired set.
\end{proof}

\begin{figure}[h]
%\vspace*{-3mm}
\centering
 \resizebox*{4.2in}{!}{
 \includegraphics{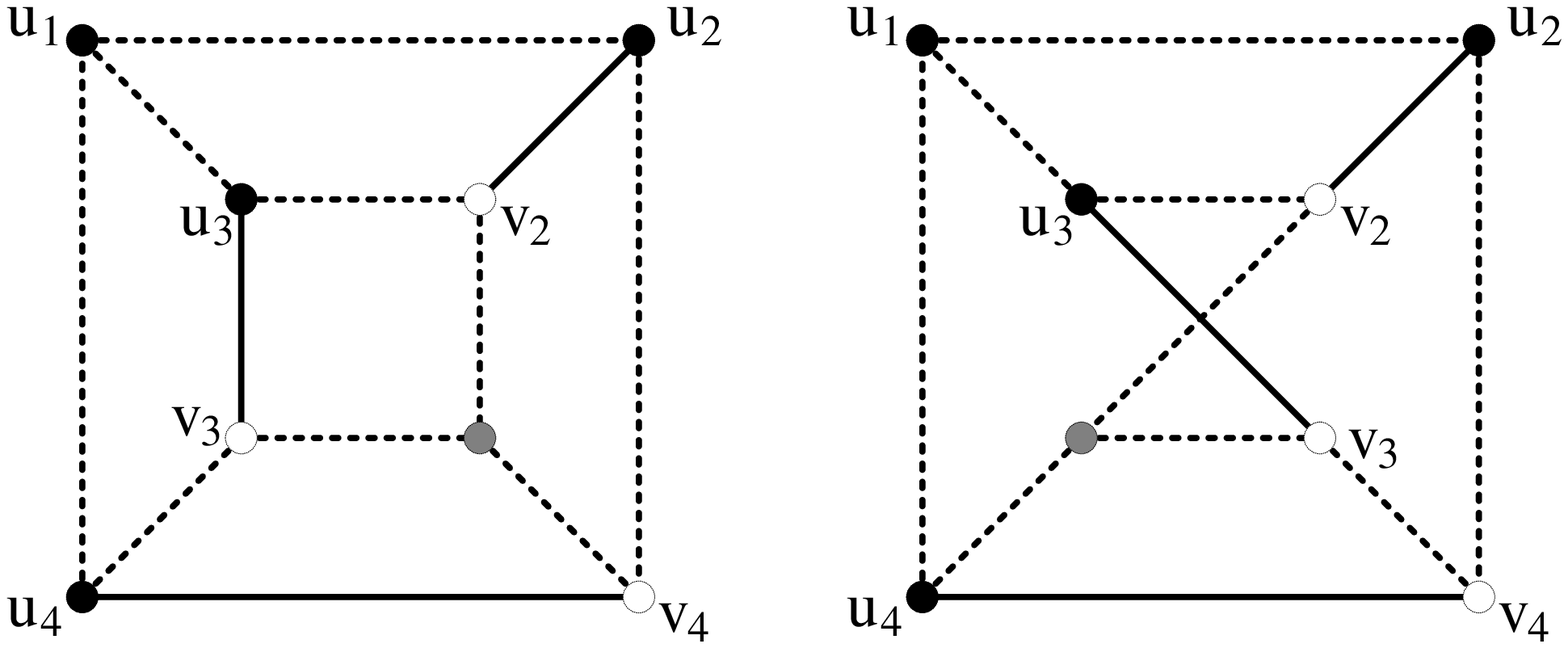}}%\vspace*{-104mm}
 \caption{Two structures of $XQ_3$ in the proof of Lemma \ref{l3}.} \label{xq3_proof}\vspace*{12mm}
 \end{figure}

\section{Local diagnosis under the BGM model}\label{sec4}

First, we establish a necessary and sufficient condition for ensuring distinguishability in the following theorem.

\begin{theorem}\label{BGMdistinguishable}
Let $F_1$ and $F_2$ be any two distinct node subsets of a system $G=(V,E)$. Thus, $(F_1,F_2)$ is a distinguishable pair if and only if one of the following states holds:

1. There exist a node $u\in V-(F_1\cup F_2)$ and a node $v\in F_1\bigtriangleup F_2$ such that $(u,v)\in E$;

2. There exist two distinct nodes $u,v\in F_1-F_2$ such that $(u,v)\in E$;

3. There exist two distinct nodes $u,v\in F_2-F_1$ such that $(u,v)\in E$. (See Figure \ref{dis} for an illustration.)
\end{theorem}

\begin{figure}[h]
\vspace*{-2mm}
 \begin{center}
 \resizebox*{4in}{!}{
 \includegraphics{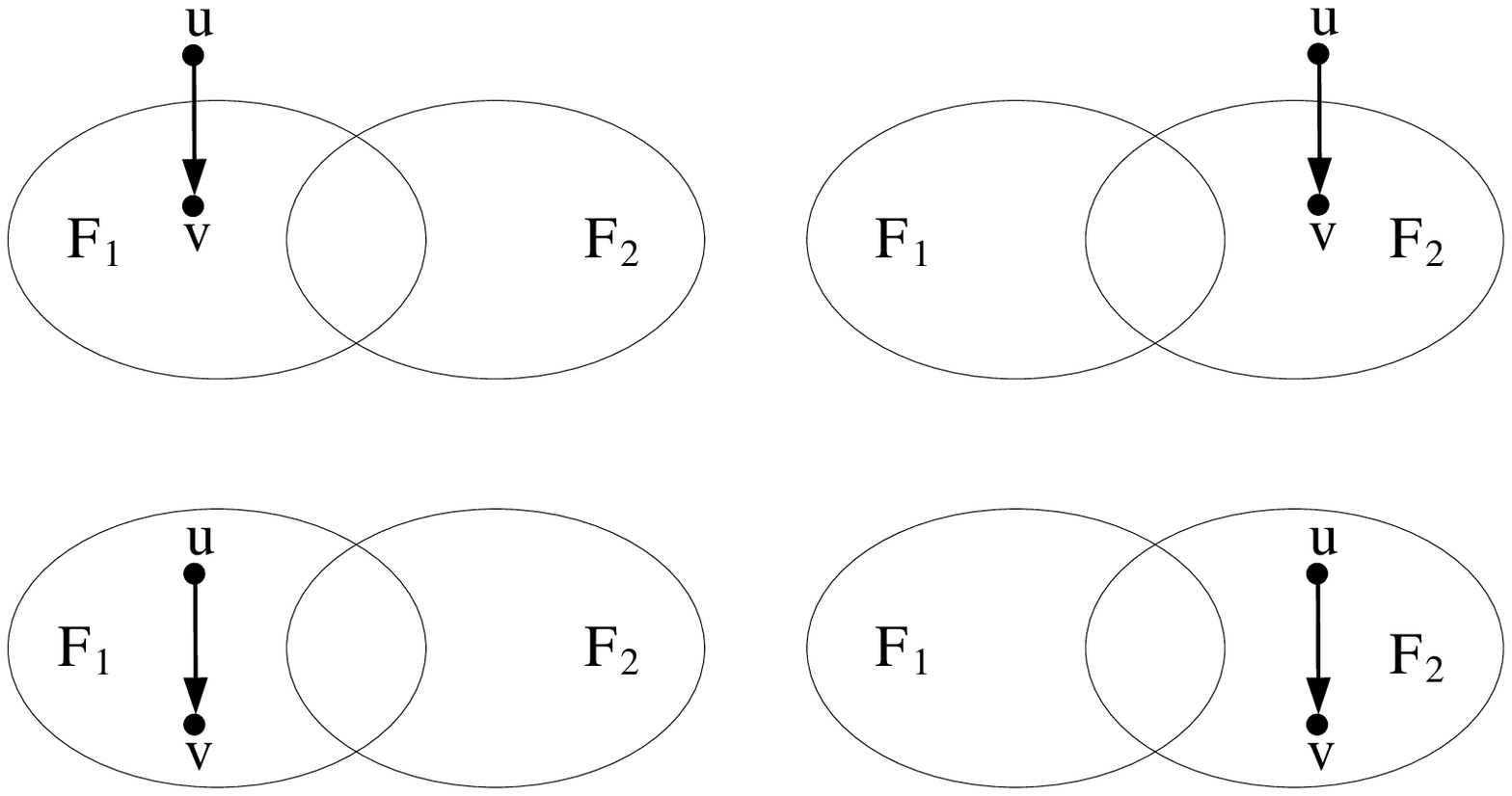}}%\vspace*{-103mm}
 \caption{An illustration of Theorem \ref{BGMdistinguishable}.}
 \label{dis}
 \end{center}\vspace*{-6mm}
\end{figure}

\begin{figure}[b]
\vspace*{-2mm}
\centering
 \resizebox*{3in}{!}{
\hspace{10mm} \includegraphics{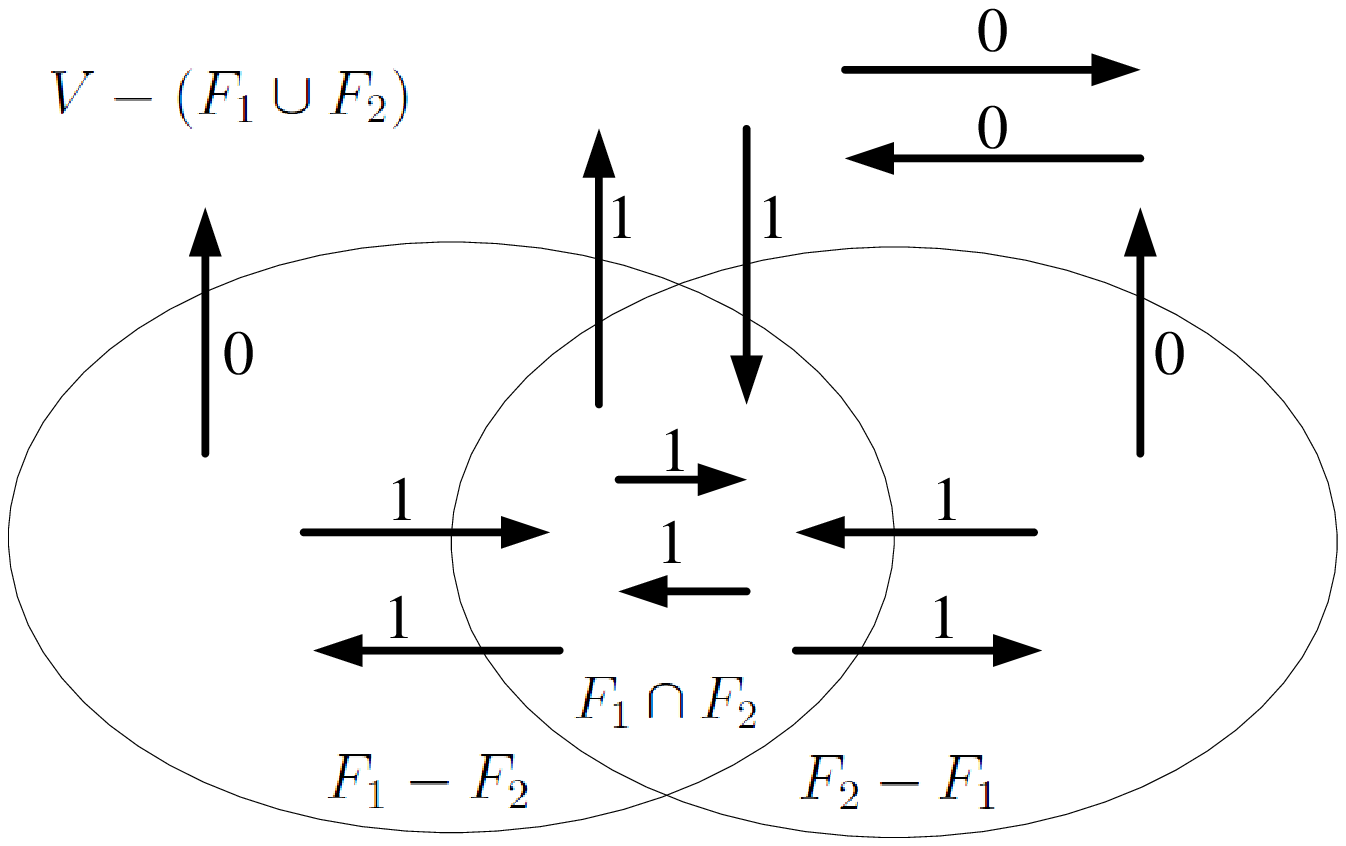}}%\vspace*{-91mm}
 \caption{A syndrome for which both $F_1$ and $F_2$ are allowable faulty sets under the BGM model.}
 \label{compare}
\end{figure}

\begin{proof}
First, we prove the necessary condition. Suppose that $(F_1,F_2)$ is a distinguishable pair, and none of the states holds. Hence, there exists a syndrome (Figure \ref{compare}) such that $F_1$ and $F_2$ are allowable faulty sets under the BGM model, thus contradicting the assumption that $(F_1,F_2)$ is a distinguishable pair.

\eject

We then prove the sufficient condition. Suppose that one of the states holds, and $(F_1,F_2)$ is an indistinguishable pair. Thus, a syndrome exists such that $F_1$ and $F_2$ are allowable faulty sets under the BGM model. Without loss of generality, we assume that $u\in V-(F_1\cup F_2)$ and $v\in F_1-F_2$ exist. Under the BGM model, if $\sigma(u,v)=0$, $F_1$ is not the faulty set. If $\sigma(u,v)=1$, $F_2$ is not the faulty set. We then consider another two states. Without loss of generality, assume that $u,v\in F_1-F_2$. Under the BGM model, if $\sigma(u,v)=0$, $F_1$ is not the faulty set. If $\sigma(u,v)=1$, $F_2$ is not the faulty set. Thus, the allowable faulty set is unique, and $(F_1,F_2)$ is a distinguishable pair, thereby contradicting the assumption that $(F_1,F_2)$ is an indistinguishable pair.
\end{proof}

For the local diagnosability, we have the following theorem.

\begin{theorem}\label{BGMdis}
Let $G=(V,E)$ be a system, and $u\in V(G)$. $G$ is locally $t$-diagnosable at node $u$ if and only if for any two distinct sets $F_1,F_2\subset V$ with $|F_1|\leq t$, $|F_2|\leq t$ and $u\in F_1\bigtriangleup F_2$, $(F_1,F_2)$ is a distinguishable pair.
\end{theorem}

Here, we propose a specific structure called the $t$-diagnosis-tree for local diagnosis under the BGM model. The definition of a $t$-diagnosis-tree is as follows.

\begin{definition}\label{dt}
A $t$-diagnosis-tree $DT_{t}(u)$ is a tree with order $t$ and rooted at $u$, such that $V(DT_{t}(u))=\{u\}\cup \{x_{i}\mid 1\leq i\leq t\}\cup \{y_{i}\mid 1\leq i\leq t\}$, and $E(DT_{t}(u))=\{\{u,x_{i}\},\{x_{i},y_{i}\}\mid 1\leq i\leq t\}$. Figure~\ref{trtree} illustrates the $DT_{t}(u)$.
\end{definition}

\begin{figure}[h]
\vspace*{-4mm}
 \begin{center}
 \resizebox*{1.5in}{!}{
 \includegraphics{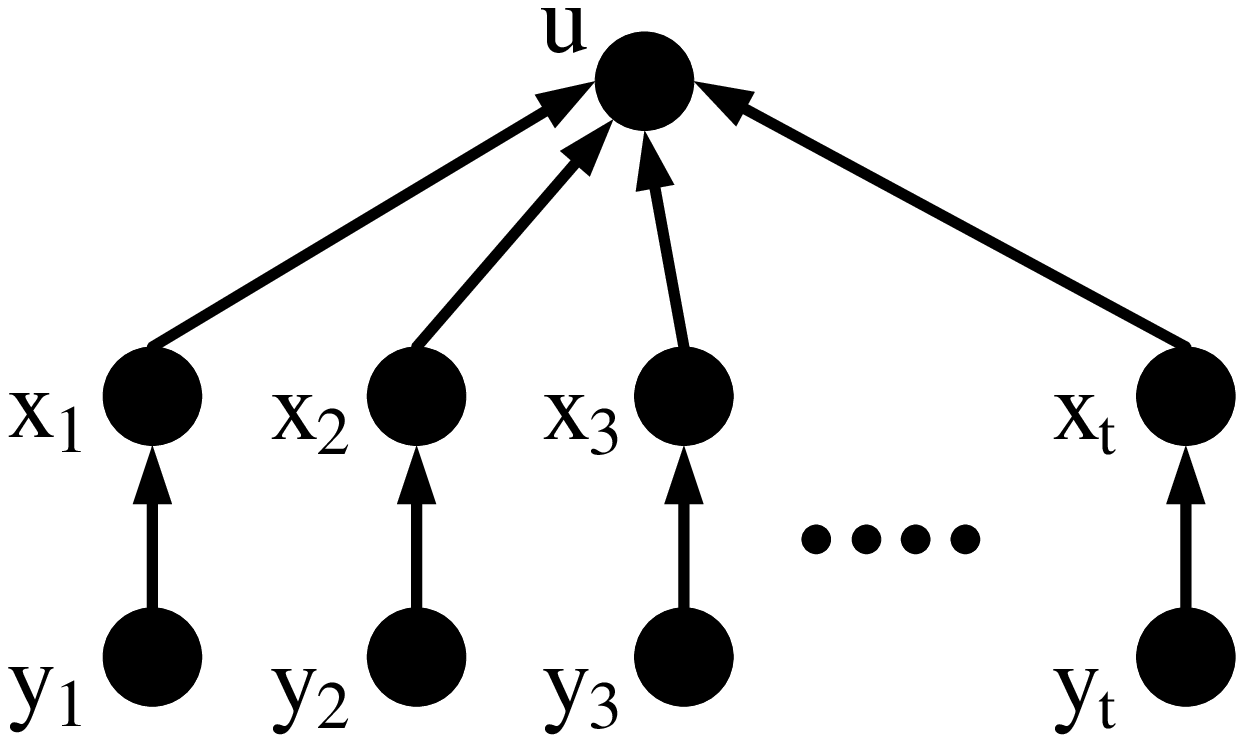}}%\vspace*{-75mm}
 \caption{A $t$-diagnosis-tree $DT_{t}(u)$ of Definition \ref{dt}.}
 \label{trtree}
 \end{center}\vspace*{-2mm}
\end{figure}

We propose the local diagnosis algorithm (\textbf{LDA}) for determining the fault status of a node $u$ in a $t$-diagnosis-tree $DT_{t}(u)$ with two test rounds under the BGM model in Algorithm \ref{two_rounds}.

\begin{algorithm} \label{two_rounds}
    \caption{{\bf LDA}$(DT_{t}(u))$}
    \KwIn{A $t$-diagnosis-tree $DT_{t}(u)$.}
    \KwOut{The value is $0$ or $1$ if $u$ is fault-free or faulty, respectively.}
%    \SetVline
    \Begin{
        {\bf Round 1:} Perform the test $(y_{i},x_{i})$ for every $1\leq i\leq t$.

        {\bf set} $A=\{x_{i}\mid \sigma(y_{i},x_{i})=0$ for $1\leq i\leq t\}$\;

        \If{$|A|\neq 0$}{
            {\bf set} $z$ being a node in $A$\;

            {\bf Round 2:} Perform the test $(z,u)$.

                \lIf{$\sigma(z,u)=0$}{\Return $0$}

                \lElse{\Return $1$}
        }
        \Else{
            \Return $0$\;
        }
    }
\end{algorithm}

\medskip
We then prove that a node $u$ in a $t$-diagnosis-tree $DT_{t}(u)$ can be diagnosed accurately in two test rounds with \textbf{LDA}$(DT_{t}(u))$ under the BGM diagnosis model.

\begin{theorem}\label{lda_proof}
Suppose that $DT_{t}(u)$ is a $t$-diagnosis-tree with order $t$ and rooted at $u$. If $F$ is a faulty set in $DT_{t}(u)$ with $|F|\leq t$, then the faulty/fault-free status of $u$ can be identified accurately in two test rounds with \textbf{LDA}$(DT_{t}(u))$ under the BGM diagnosis model.
\end{theorem}

\begin{proof}
Depending on the definition of the BGM model and the results listed in Table \ref{BGMtable}, if $\sigma(y_{i},x_{i})=0$ for every $1\leq i\leq t$, the node $x_i$ is fault-free. Thus, the test $(x_i,u)$ is reliable. If $\sigma(y_{i},x_{i})=1$ for every $1\leq i\leq t$, there exists at least one faulty node in $\{x_i,y_i\}$. Assume that $u\in F$. We have $|F|\geq t+1$, which contradicts the assumption that $|F|\leq t$. Thus, the theorem holds.
\end{proof}

%\vspace*{2mm}
\section{One good neighbor conditional local diagnosis algorithm under the BGM model}\label{sec5}

Lai et al. proposed the concept of conditional fault diagnosis by restricting that, for each processor $v$ in the network, all processors directly connected to $v$ do not fail simultaneously \cite{Lai05}. Suppose that $G=(V,E)$. A set $F\subset V(G)$ is a {\it conditional faulty set} if $N_{G}(v)\not\subseteq F$ for any node $v\in V(G)-F$. A system $G$ is {\it conditionally faulty} if the faulty node set of $G$ forms a conditional faulty set. For any two distinct conditional faulty sets $F_1$ and $F_2$ of $G$ with $|F_1| \leq t$ and $|F_2| \leq t$, if $(F_1,F_2)$ is a distinguishable pair, $G$ is {\it conditionally $t$-diagnosable}. The maximum number of conditional faulty nodes that can be accurately identified in $G$ is called the {\it conditional diagnosability} of $G$.

\begin{theorem}
Suppose that $F_1$ and $F_2$ are two distinct conditional faulty sets in a system $G$ with $|F_1| \leq t$ and $|F_2| \leq t$. Thus, $G$ is conditionally $t$-diagnosable under the BGM diagnosis model.
\end{theorem}

\begin{proof}
Without loss of generality, consider the faulty set $F_1$. Let $u\in V(G)-F_1$. Because $F_1$ is a conditional faulty set, a node $v$ is adjacent to $u$ such that $v\notin F_1$. According to Theorems \ref{BGMdistinguishable} and \ref{BGMdis}, $G$ is conditionally $t$-diagnosable under the BGM diagnosis model.
\end{proof}

Here, we propose the conditional local diagnosis algorithm (\textbf{CLDA}) for determining the fault status of a node $u$ in a conditional faulty diagnostic system $G$ under the BGM model in Algorithm \ref{CLDAalgo}.

\begin{algorithm} \label{CLDAalgo}
    \caption{{\bf CLDA}$(G,u,t)$}
    \KwIn{A graph $G$ and a node $u\in V(G)$ with $\textrm{deg}_{G}(u)=t$.}
    \KwOut{The value is $0$ or $1$ if $u$ is fault-free or faulty, respectively.}
%    \SetVline
    \Begin{
        Perform the tests $(u,v_{i})$ and $(v_{i},u)$, where $v_i\in N_{G}(u)$ for every $1\leq i\leq t$.

        \If{there exists at least one test result pair $((u,v_i),(v_i,u))=(0,0)$}{
            \Return $0$\;
        }
        \Else{
            \Return $1$\;
        }
    }
\end{algorithm}

%\medskip
We then prove that a node $u$ in a conditional faulty diagnostic system $G$ can be diagnosed accurately with \textbf{CLDA}$(G,u,t)$ under the BGM diagnosis model.

\begin{theorem}\label{clda_proof}
Suppose that $u$ is a node in $G$ with $\textrm{deg}_{G}(u)=t$. If $F$ is a conditional faulty set in $G$ with $|F|\leq t$, then the faulty/fault-free status of $u$ can be identified accurately with \textbf{CLDA}$(G,u,t)$ under the BGM diagnosis model.
\end{theorem}

\begin{proof}
Depending on the definition of the BGM model and the results listed in Table \ref{BGMtable}, if $u\in F$, test result pair $((u,v_i),(v_i,u))\in \{(0,1),(1,0),(1,1)\}$, where $v_i\in N_{G}(u)$ for every $1\leq i\leq t$. Suppose that $u\notin F$. According to the definition of a conditionally faulty system, there exists at least one node $v_i\in N_{G}(u)$ for some $1\leq i\leq t$ such that $v_i\notin F$. Thus, we have the test result pair $((u,v_i),(v_i,u))=(0,0)$ and the theorem holds.
\end{proof}

\section{A 3-round diagnosis algorithm of hypercube-like networks under the BGM model}\label{sec6}

In this section, we prove that the nodes in an $n$-dimensional hypercube-like network $XQ_n$ can be diagnosed correctly with a faulty set $F$ in three test rounds under the BGM diagnosis model if $|F|\leq n$.

\begin{theorem}\label{main_algo_proof}
Let $XQ_n$ be an $n$-dimensional hypercube-like graph. If $F$ is a faulty set in $XQ_n$ with $|F|\leq n$, then the faulty/fault-free nodes of $XQ_n$ can be identified correctly in three test rounds under the BGM diagnosis model.
\end{theorem}
\begin{proof}
Let the edges in $XQ_n$ labeled by $n$ be the perfect matching. We give the algorithm {\bf DHL} (Diagnosis for Hypercube-like graph) in Algorithm \ref{main_algo} to identify the faulty/fault-free status of the nodes in $XQ_n$ under the BGM diagnosis model. In the first test round, let the nodes in $XQ_{n-1}^{0}$ test the nodes in $XQ_{n-1}^{1}$. We set $A=\{v_i\mid \sigma(\bar{v_i},v_i)=1,$ where $\bar{v_i}\in V(XQ_{n-1}^{0})$ and $v_i\in V(XQ_{n-1}^{1})\}$. We consider the following cases.

%%%%%%%%%%%%%%%%%%%%%%%%%%%%%%%%%%%%%%%%%%%%%%%%%%%%%%%%%%%%%%%%%%%%
\begin{algorithm}\small \label{main_algo}
    \caption{{\bf DHL}$(XQ_n)$}
    \KwIn{A hypercube-like network $XQ_n$.}
    \KwOut{The node in $XQ_n$ is fault-free or faulty.}
%    \SetVline
    \Begin{

        The nodes in $XQ_{n-1}^{0}$ test the nodes in $XQ_{n-1}^{1}$.  /* Test round 1. */

        {\bf set} $A=\{v_i\mid \sigma(\bar{v_i},v_i)=1,$ where $\bar{v_i}\in V(XQ_{n-1}^{0})$ and $v_i\in V(XQ_{n-1}^{1})\}$.

        \If{$A$ is not isomorphic to $K_{1,n-1}$}{

            \lIf{$|A|\leq n-1$}{\Return{{\bf DHLA}$(XQ_n,A)$}}
            \lElse{\Return{\bf DHLB}$(XQ_n,A)$}
        }
        \Else{
            {\bf set} $x$ being the node in $A$ such that $N_{XQ_{n-1}^{1}}(x)=A-\{x\}$.

            {\bf set} $C'$ being the nodes in $V(XQ_{n-1}^{1})-A$ such that $A-\{x\}$ and $C'$ form a perfect matching in $XQ_{n-1}^{1}$.

            \Return{\bf DHLC}$(XQ_n,A,x,C')$
        }

    }
\end{algorithm}

\begin{algorithm}\small \label{main_algoA}
    \caption{{\bf DHLA}$(XQ_n,A)$}
    \KwIn{A hypercube-like network $XQ_n$. A set $A=\{v_i\mid \sigma(\bar{v_i},v_i)=1,$ where $\bar{v_i}\in V(XQ_{n-1}^{0})$ and $v_i\in V(XQ_{n-1}^{1}),1\leq i\leq n\}$ with $|A|\leq n-1$.}
    \KwOut{The node in $XQ_n$ is fault-free or faulty.}
%    \SetVline
    \Begin{

        The nodes in $XQ_{n-1}^{1}$ test the nodes in $XQ_{n-1}^{0}$. /* Test round 2. */

        {\bf set} $B=\{\bar{v_i}\mid \sigma(v_i,\bar{v_i})=1,$ where $\bar{v_i}\in V(XQ_{n-1}^{0})$ and $v_i\in V(XQ_{n-1}^{1})\}$.

        {\bf set} $C$ being the nodes in $V(XQ_{n-1}^{1})-A$ such that $A$ and $C$ form a perfect matching in $XQ_{n-1}^{1}$.

        \lIf{$|A|=0$}{The nodes in $B$ are faulty, and the others are fault-free}

        \ElseIf{$1\leq |B|\leq n-1$}{

            {\bf set} $D$ being the nodes in $V(XQ_{n-1}^{0})-B$ such that $B$ and $D$ form a perfect matching in $XQ_{n-1}^{0}$.

            The nodes in $C$ test the nodes in $A$, and the nodes in $D$ test the nodes in $B$. /* Test round 3. */

            {\bf set} $A_1=\{v_i\mid \sigma(c_i,v_i)=1,$ where $c_i\in C$ and $v_i\in A\}$.

            {\bf set} $B_1=\{\bar{v_i}\mid \sigma(d_i,\bar{v_i})=1,$ where $d_i\in D$ and $\bar{v_i}\in B\}$.

            \Return The nodes in $A_1$ or $B_1$ are faulty, and the others are fault-free.

        }

        \Else{The nodes in $C$ test the nodes in $A$. /* Test round 3. */

            {\bf set} $A_1=\{v_i\mid \sigma(c_i,v_i)=1,$ where $c_i\in C$ and $v_i\in A\}$.

            \Return The nodes in $A_1$ are faulty; for every node $\bar{v_i}\in B$, $\bar{v_i}$ is faulty if $v_i$ is not in $A_1$; the others are fault-free.

        }
    }
\end{algorithm}

\eject
{\bf Case 1:} Suppose that $A$ is not isomorphic to $K_{1,n-1}$. We consider the following cases. \medskip

{\bf Subcase 1.1:} Suppose that $|A|\leq n-1$. We give the algorithm {\bf DHLA} in Algorithm \ref{main_algoA} to identify the faulty/fault-free status of the nodes in $XQ_n$. In the second test round, let the nodes in $XQ_{n-1}^{1}$ test the nodes in $XQ_{n-1}^{0}$. We set $B=\{\bar{v_i}\mid \sigma(v_i,\bar{v_i})=1,$ where $\bar{v_i}\in V(XQ_{n-1}^{0})$ and $v_i\in V(XQ_{n-1}^{1})\}$. By Lemma \ref{l1}, there exists a node set $C$ in $XQ_{n-1}^{1}-A$ such that $A$ and $C$ form a perfect matching in $XQ_{n-1}^{1}$. We consider the following cases.

\medskip
{\bf Subcase 1.1.1:} Suppose that $|A|=0$. Under the BGM diagnosis model, the nodes in $B$ are faulty, and the others are fault-free.

\medskip
{\bf Subcase 1.1.2:} Suppose that $1\leq |A|\leq n-1$ and $1\leq |B|\leq n-1$. By Lemma \ref{l1}, there exists a node set $D$ in $XQ_{n-1}^{0}-B$ such that $B$ and $D$ forms a perfect matching in $XQ_{n-1}^{0}$. Under the BGM diagnosis model, the nodes in $C\cup D$ is fault-free. In the third test round, let the nodes in $C$ test the nodes in $A$, and let the nodes in $D$ test the nodes in $B$. We set $A_1=\{v_i\mid \sigma(c_i,v_i)=1,$ where $c_i\in C$ and $v_i\in A\}$ and $B_1=\{\bar{v_i}\mid \sigma(d_i,\bar{v_i})=1,$ where $d_i\in D$ and $\bar{v_i}\in B\}$. Since the nodes in $C\cup D$ is fault-free, the tests performed by the nodes in $C\cup D$ are reliable. Thus the nodes in $A_1$ or $B_1$ are faulty, and the others are fault-free. \medskip

{\bf Subcase 1.1.3:} Suppose that $1\leq |A|\leq n-1$ and $|B|=n$. In the third test round, let the nodes in $C$ test the nodes in $A$. We set $A_1=\{v_i\mid \sigma(c_i,v_i)=1,$ where $c_i\in C$ and $v_i\in A\}$. Since $|B|=n$, $F\subseteq A\cup B$. The tests performed by the nodes in $C$ are reliable. Thus the nodes in $A_1$ are faulty; for every node $\bar{v_i}\in B$, $\bar{v_i}$ is faulty if $v_i$ is not in $A_1$; the others are fault-free. \medskip

\begin{algorithm}[!b]\label{main_algoB}
    \caption{{\bf DHLB}$(XQ_n,A)$}
    \KwIn{A hypercube-like network $XQ_n$. A set $A=\{v_i\mid \sigma(u_i,v_i)=1,$ where $u_i\in V(XQ_{n-1}^{0})$ and $v_i\in V(XQ_{n-1}^{1})\}$ with $|A|=n$.}
    \KwOut{The node in $XQ_n$ is fault-free or faulty.}
%    \SetVline
    \Begin{
            {\bf set} $C$ being the nodes in $V(XQ_{n-1}^{1})-A$ such that $A$ and $C$ form a perfect matching in $XQ_{n-1}^{1}$.

            The nodes in $C$ test the nodes in $A$. /* Test round 2. */

            {\bf set} $A_0=\{v_i\mid \sigma(c_i,v_i)=0,$ where $c_i\in C$ and $v_i\in A\}$.

            {\bf set} $A_1=\{v_i\mid \sigma(c_i,v_i)=1,$ where $c_i\in C$ and $v_i\in A\}$.

            \Return The nodes in $A_1$ are faulty; for every node $\bar{v_i}$, $\bar{v_i}$ is faulty if $v_i$ is in $A_0$; the others are fault-free.
    }
\end{algorithm}

{\bf Subcase 1.2:} Suppose that $|A|=n$. We give the algorithm {\bf DHLB} in Algorithm \ref{main_algoB} to identify the faulty/fault-free status of the nodes in $XQ_n$. By Lemma \ref{l2}, there exists a node set $C$ in $XQ_{n-1}^{1}-A$ such that $A$ and $C$ forms a perfect matching in $XQ_{n-1}^{1}$. In the second test round, let the nodes in $C$ test the nodes in $A$. We set $A_0=\{v_i\mid \sigma(c_i,v_i)=0,$ where $c_i\in C$ and $v_i\in A\}$ and $A_1=\{v_i\mid \sigma(c_i,v_i)=1,$ where $c_i\in C$ and $v_i\in A\}$. Since $|A|=n$, $F\subseteq A\cup B$. The tests performed by the nodes in $C$ are reliable. Thus the nodes in $A_1$ are faulty; for every node $\bar{v_i}$, $\bar{v_i}$ is faulty if $v_i$ is in $A_0$; the others are fault-free. \medskip

\begin{algorithm}[!b]\label{main_algoC}
    \caption{{\bf DHLC}$(XQ_n,A,x,C')$}
    \KwIn{A hypercube-like network $XQ_n$. A set $A=\{v_i\mid \sigma(u_i,v_i)=1,$ where $u_i\in V(XQ_{n-1}^{0})$ and $v_i\in V(XQ_{n-1}^{1})\}$ with $A$ being isomorphic to $K_{1,n-1}$. A node $x\in A$ with $N_{XQ_{n-1}^{1}}(x)=A-\{x\}$. A set $C'\in V(XQ_{n-1}^{1})-A$, where $A-\{x\}$ and $C'$ form a perfect matching in $XQ_{n-1}^{1}$.}
    \KwOut{The node in $XQ_n$ is fault-free or faulty.}
%    \SetVline
    \Begin{

            The nodes in $C'$ test the nodes in $A-\{x\}$. /* Test round 2. */

            {\bf set} $A_0=\{v_i\mid \sigma(c_i,v_i)=0,$ where $c_i\in C'$ and $v_i\in A-\{x\}\}$.

            {\bf set} $A_1=\{v_i\mid \sigma(c_i,v_i)=1,$ where $c_i\in C'$ and $v_i\in A-\{x\}\}$.

            \If{$|A_0|\neq 0$}{

                {\bf set} $y$ being a node in $A_0$.

                The node $y$ tests the node $x$. /* Test round 3. */

                \If{$\sigma(y,x)=0$}{
                     \Return $\bar{x}$ is faulty; the nodes in $A_1$ are faulty; for every node $\bar{v_i}$, $\bar{v_i}$ is faulty if $v_i$ is in $A_0$; the others are fault-free.
                }
                \Else{
                    \Return $x$ is faulty; the nodes in $A_1$ are faulty; for every node $\bar{v_i}$, $\bar{v_i}$ is faulty if $v_i$ is in $A_0$; the others are fault-free.
                }

            }
            \Else{
                {\bf set} $z$ being a node in $XQ_{n-1}^{0}-\{\bar{x}\}$.

                The node $z$ tests the node $\bar{x}$. /* Test round 3. */

                \If{$\sigma(z,\bar{x})=0$}{
                     \Return $x$ is faulty; the nodes in $A_1$ are faulty; for every node $\bar{v_i}$, $\bar{v_i}$ is faulty if $v_i$ is in $A_0$; the others are fault-free.
                }
                \Else{
                    \Return $\bar{x}$ is faulty; the nodes in $A_1$ are faulty; for every node $\bar{v_i}$, $\bar{v_i}$ is faulty if $v_i$ is in $A_0$; the others are fault-free.
                }
            }

    }
\end{algorithm}

{\bf Case 2:} Suppose that $A$ is isomorphic to $K_{1,n-1}$. Let $x$ be the node in $A$ such that $N_{XQ^{1}_{n-1}}(x)=A-\{x\}$. By Lemma \ref{l3}, there exists a node set $C'$ in $V(XQ_{n-1}^{1})-A$ such that $A-\{x\}$ and $C'$ form a perfect matching in $XQ_{n-1}^{1}$. We give the algorithm {\bf DHLC} in Algorithm \ref{main_algoC} to identify the faulty/fault-free status of the nodes in $XQ_n$. In the second test round, let the nodes in $C'$ test the nodes in $A-\{x\}$. Since $|A|=n$, $F\subseteq A\cup B$. The tests performed by the nodes in $C'$ are reliable. Thus the nodes in $A_1$ are faulty; for every node $\bar{v_i}$, $\bar{v_i}$ is faulty if $v_i$ is in $A_0$; the nodes in $A_0$ are fault-free. For the faulty/fault-free status of the node $x$, we consider the following cases.

\medskip
{\bf Subcase 2.1:} Suppose that $|A_0|\neq 0$. Let $y$ be a node in $A_0$. In the third test round, let $y$ test $x$. If $\sigma(y,x)=0$, $x$ is fault-free, and $\bar{x}$ is faulty. If $\sigma(y,x)=1$, $x$ is faulty, and $\bar{x}$ is fault-free.

\medskip
{\bf Subcase 2.2:} Suppose that $|A_0|=0$. Let $z$ be a node in $V(XQ^{0}_{n-1})-\{\bar{x}\}$. In the third test round, let $z$ test $\bar{x}$. Since $F\subseteq A\cup B$, the test performed by $z$ is reliable. If $\sigma(z,\bar{x})=0$, $\bar{x}$ is fault-free, and $x$ is faulty. If $\sigma(z,\bar{x})=1$, $\bar{x}$ is faulty, and $x$ is fault-free.
\end{proof}

\section{Concluding remarks}\label{sec7}

In this paper, a diagnosis testing signal is supposed to be delivered from one node to another node through the communication bus at one time. The node is not allowed to perform multiple tests simultaneously. There are many paired tests that can be performed parallel in a test round. Each node can only have one of the following state in a round, testing, being tested, and not participating in any testing. In \cite{Teng2014}, Teng and Lin discussed the local diagnosability of a $t$-diagnosable system under the PMC diagnosis model. They proved that any reliable diagnosis algorithm should be completed in at least three test rounds under the PMC model. The PMC model is a more general diagnosis model, and the BGM model can be regarded as a special case of the PMC model. The definitions of the two diagnosis models are different. In this paper, we propose an algorithm for determining the fault status of a node in a $t$-diagnosable system with the structure $t$-diagnosis-tree under the BGM model, and we prove that the diagnosis can be completed in two test rounds. We also give an algorithm for conditional local diagnosis under the BGM model. The structure $t$-diagnosis-tree can be embedded in many well-known interconnection networks of multiprocessor systems; for instance, hypercubes, star graphs, and arrangement graphs. We give an algorithm for determining the fault status of a node in a hypercube-like network. Suppose that $XQ_n$ is an $n$-dimensional hypercube-like network, and $F$ is a faulty set in $XQ_n$ with $|F|\leq n$. We prove that the fault status of nodes in $XQ_n$ can be identified in three test rounds under the BGM diagnosis model. To perform our algorithm, we have to find the perfect matching in the hypercube-like graph. Finding a perfect matching can be solved in polynomial time by the algorithm of Edmonds \cite{Edmonds65}. Therefore, with our algorithm, the diagnosis can be completed in polynomial time. Future research will endeavor to determine specific structures for existing practical interconnection networks, design an efficient diagnosis algorithm, and prove the diagnosability of this useful structure under the BGM diagnosis model.

\section*{Acknowledgements}
This work was supported in part by the Ministry of Science and Technology of the Republic
of China under Contract MOST 109-2221-E-126-004.

%%%%%%%%%%%%%%%%%%%%%%%%%%%%%%%%%%%%%%%%%%%%%%%%%%%%%%%%%%%%%%%%%%%%%%

\end{document}